\documentclass[11pt]{article}

\usepackage{amsmath,amssymb}
\usepackage{physics}
\usepackage{hyperref}
\usepackage{cite}
\usepackage[utf8]{inputenc}
\usepackage{amsthm}
\usepackage{booktabs}
\usepackage{graphicx}
\usepackage{float} 

\newtheorem{theorem}{Theorem}
\newtheorem{proposition}{Proposition}

\newtheorem{remark}{Remark}

\title{Quantum-Like Models of Cognition and Decision Making:\\
Open-Systems and Gorini--Kossakowski--Sudarshan--Lindblad Dynamics}

\author{Masanari Asano$^{1}$, Irina Basieva$^{2},$ and Andrei Khrennikov$^{3,*}$}
\date{}

\begin{document}
\maketitle
\noindent
$^{1}$Faculty of Science and Technology, Information Sciences, Tokyo University of Science\\
2641 Yamazaki, Noda-shi, Chiba-ken 278-8510, Japan\\
$^{2}$Quantum Group, School of Computing, Newcastle University\\ 
1 Science Square, Newcastle upon Tyne, NE4 5TG, UK\\
$^{3,*}$Center for Mathematical Modeling in Physics and Cognitive Sciences\\
Linnaeus University, V\"axj\"o, SE-351 95, Sweden\\
*Corresponding author email: Andrei.Khrennikov@lnu.se

\begin{abstract}
This paper starts with surveying the evolution of quantum-like models of cognition and decision making, transitioning from static kinematic representations to a robust dynamical framework based on open quantum systems. We provide a comprehensive analysis of the Gorini-Kossakowski-Sudarshan-Lindblad (GKSL) master equation's application in cognitive psychology and decision making, illustrating how it models mental state evolution as a dissipative process influenced by an informational environment. We categorize dynamical regimes into Passive and Active Hamiltonians, demonstrating how non-commutation with projections on decision basis serves as a mathematical signature of cognitive agency and Quantum Escape from classical equilibria. The utility of this framework is further explored through its ability to stabilize non-Nash outcomes in strategic games, such as the Prisoner's Dilemma. Building upon this dynamical foundation, we identify ``cognitive beats'' as a signature of the internal struggle between competing ``flows of mind'' deliberated at approximately equal frequencies. Distinct from the damped oscillations of simple interference, these beats emerge from a structural tension between Liouvillian channels that generates a secondary, slow-scale modulation of conviction. This beat envelope dictates the timing of peak readiness and hesitation, providing a mathematical map of the transition between conflicting cognitive states. By resolving these nested time scales, we provide a new spectral diagnostic for the depth of cognitive agency and the complexity of the underlying deliberation process. This paper develops a theoretical framework linking GKSL dynamics with quantum-like cognition and decision-making (QCDM), highlighting how dissipative quantum models can capture features of human thought and decision processes.
\end{abstract}


\section{Introduction}

The field of {\it Quantum-Like Modeling} (QLM) has established itself as a mathematically rigorous alternative to classical Kolmogorovian probability, successfully accounting for empirical phenomena such as interference effects (violations of the law of total probability), order, conjunction, disjunction, and response replicability effects, contextuality, incompatibility of observables, e.g., questions.\footnote{Representative applications of QLMs can found in works \cite{INT}-\cite{Open_QLM}, \cite{Bruza}-\cite{Busemeyer5}, \cite{Asano1}-\cite{QL3}, \cite{Pothos1}-\cite{Pothos3}, \cite{Bagarell2}-\cite{Bagarello4},\cite{PL4,Zeit,KhrennikovaHaven2016,OK20,OK21,laser11}.} However, early iterations of these models were predominantly static, providing probability assignments for fixed moments rather than capturing the continuous, temporal nature of thought. Because human cognition is a dynamical process shaped by constant information flow, social interaction, and internal fluctuations, there is a fundamental need for a framework that accounts for mental state evolution.

While standard Schr\"odinger-type dynamics offer a starting point, their inherent reversibility and conservation of entropy fail to reflect the irreversible nature of real-world decision making. Consequently, researchers have turned to the theory of open quantum systems, specifically the {\it Gorini--Kossakowski--Sudarshan--Lindblad}  (GKSL) master equation \cite{GoriniKossakowskiSudarshan1976,Lindblad1976}, to describe cognitive agents as systems interacting with a ``psycho-mental'' or informational environment. This approach, known as \textit{decision making through decoherence}, allows for a principled transition from mental superpositions - representing competing alternatives - to stable ``steady states'' that signify a commitment to a decision.

This paper develops a theoretical framework connecting GKSL dynamics with quantum-like models of cognition and decision-making (QCDM), demonstrating how dissipative quantum dynamics can provide a rigorous mathematical foundation for capturing features of human thought and decision processes. The framework may be of interest not only to researchers in cognitive science, decision making, social science, and biology, but also to experts in the mathematical theory of GKSL dynamics, opening a new avenue for applications of this well-established formalism.\footnote{ The GKSL dynamics, rigorously formulated over 50 years ago, has found numerous applications in quantum physics, supported by the development of advanced mathematical theory. Here, we propose a novel application of this theory to QCDM and generally to QLM.}

A central advancement of the present work is the formal distinction between Passive and Active Hamiltonians within this dissipative evolution. While passive dynamics merely reflect the environmental pull, an active Hamiltonian signifies the internal agency of the decision-maker, where non-commutation between internal operators and environmental influences leads to a Quantum Escape from classical equilibrium. This mechanism provides a mathematical foundation for cognitive autonomy, allowing the agent to maintain stable preferences or reach non-classical outcomes that would be unreachable under purely Markovian classical relaxation.

This framework demonstrates significant utility in the realm of strategic interactions, particularly in resolving paradoxes within game theory. By modeling the Prisoner’s Dilemma through GKSL dynamics, we show that the interplay between decoherence and coherent internal updates can stabilize non-Nash cooperative outcomes. This suggests that selection of cooperative strategies is not merely a deviation from rationality, but a stable dynamical state emerging from the continuous interaction between the agent's internal cognitive state and the social informational environment.

Finally, we explore a sophisticated dynamical regime characterized by ``cognitive beats.'' These beats represent an internal struggle between competing ``flows of mind''---deliberative pathways that operate at nearly identical frequencies. Distinct from the simple damped oscillations of standard interference, these beats arise from a structural tension between Liouvillian channels, manifesting as a slow-scale modulation of conviction. By resolving these nested time scales, we can map the transition between peak readiness and hesitation, offering a new spectral diagnostic for the complexity and depth of the underlying deliberation process.

In connection with the cognitive beats phenomenon, we refer to the article \cite{eye}, where eye-tracking measurements  were used as a signature of convergence toward a steady state. In that work, fluctuating damping behavior in eye-tracking data was experimentally observed, and the data were simulated using GKSL dynamics within a decoherence-based framework of decision making. The graphs presented in \cite{eye} exhibit fluctuations; however, the sample size is not sufficient to determine whether cognitive beats are present.  

\section{Static quantum-like scheme for decision making}

The cognitive state of an agent (decision maker) $S$ is represented by a quantum state, mathematically described as a density operator acting on a complex Hilbert space $\mathcal{H}$. 
In the early formulations of QCDM, and more generally in QLM, only pure states were typically considered. 
However, the adoption of concepts from the theory of open quantum systems has motivated the systematic use of mixed states as well.

Cognitive observables—such as questions, tasks, or psychological traits—are represented by Hermitian linear operators $  A : \mathcal{H} \to \mathcal{H}$. 
The spectrum of the operator $  A$ corresponds to the set of possible outcomes of the associated cognitive measurement. 
In the finite-dimensional case, which we assume throughout this work in accordance with standard practice in QLM, the spectrum is discrete and consists of eigenvalues $\alpha_1, \ldots, \alpha_n$. 
These eigenvalues label the possible responses to the question represented by the observable $  A$.

These eigenstates correspond to the actions of responses themselves, 
which are events the agent must manifest in the natural world.
From the perspective of classical probability theory, they must be mutually exclusive.

Any Hermitian operator admits a spectral decomposition,
\begin{equation}
\label{eq:spectral}
  A = \sum_{\alpha} \alpha\,   E_A(\alpha),
\end{equation}
where $  E_A(\alpha)$ denotes the orthogonal projector onto the eigenspace $\mathcal{H}_A(\alpha)$ corresponding to the eigenvalue $\alpha$.

For a pure cognitive state $|\psi\rangle \in \mathcal{H}$, the probability of obtaining the outcome $A=\alpha$ is given by Born’s rule,
\begin{equation}
\label{eq:born_pure}
P_{\psi}(A=\alpha) = \|  E_A(\alpha)\psi\|^2 
= \langle \psi |   E_A(\alpha) | \psi \rangle.
\end{equation}
A measurement yielding the result $A=\alpha$ induces a transformation (or ``back-action'') of the cognitive state,
\begin{equation}
\label{eq:proj_pure}
|\psi\rangle \longrightarrow |\psi\rangle_{A=\alpha} =
\frac{  E_A(\alpha)|\psi\rangle}{\|  E_A(\alpha)|\psi\rangle\|}.
\end{equation}
This transformation is known as the \textit{projection postulate}, formulated in its general form by L\"uders. 

For mixed states described by density operators $ \rho$, the Born rule takes the form
\begin{equation}
\label{eq:born_mixed}
P_{ \rho}(A=\alpha) = \mathrm{Tr}( \rho\,   E_A(\alpha)).
\end{equation}
Conditioned on obtaining the outcome $A=\alpha$, the post-measurement state is given by
\begin{equation}
\label{eq:proj_mixed}
 \rho \longrightarrow  \rho_{\alpha} =
\frac{  E_A(\alpha)\,  \rho\,   E_A(\alpha)}
{\mathrm{Tr}\!\left[  E_A(\alpha)\,  \rho\,   E_A(\alpha)\right]}.
\end{equation}

The above construction follows the textbook theory of quantum measurements. 
While mathematically consistent, this framework remains largely formal and does not by itself clarify the internal structure of cognitive measurements. 
One particularly important feature is the reliance on the projection postulate. 
In QCDM, especially in approaches based predominantly on pure states, it is common to interpret this postulate as a ``collapse of the mental wave function.'' 
Superpositions are then viewed as superpositions of competing cognitive alternatives within the agent $S$.

Answering a question represented by an observable $  A$ with outcome $A=\alpha$ is modeled as a collapse of the pure cognitive state, corresponding to a projection onto the subspace $\mathcal{H}_A(\alpha)$. 
For observables with non-degenerate spectra, this reduces to a projection onto the corresponding eigenvector. 
In the QCDM literature, such decision processes are often illustrated by geometric pictures of state-vector projections in a plane (see, e.g.,~\cite{Busemeyer0,Busemeyer1}). However, this interpretation introduces a degree of conceptual mystification, both in quantum cognition and in quantum physics itself. We argue that it is desirable to eliminate this ambiguous notion from QCDM and, more generally, from QLM. 
An alternative approach that avoids explicit appeal to collapse is known as decision making through decoherence : 

Generally, the agent can be conscious all of the possible reactions. 
That is, the representations of these events exist simultaneously and interfere each other in its mind. 
This representational simultaneity is the property characterizing the agent's cognitive state,
and the intensity of interference is clearly expressed in the non-diagonal part of $ \rho$ in the basis of the eigenvectors.
On the other hand, the classical probability theory values the exclusivity of objective events and the frequencies of them, 
and in this sense, the representational simultaneity is completely vanished.
This is expressed in $ \rho_\alpha$.
The states of $ \rho$ and $ \rho_\alpha$ are semantically different: 
The cognitive measurement $ \rho \longrightarrow  \rho_{\alpha}$ is a sort of semantic reduction.

\section{The GKSL Equation in Quantum-like Cognition and Decision Making}

\subsection{Motivation and Conceptual Framework}

In our foundational work \cite{Asano1}, we proposed modeling cognitive evolution using the GKSL master equation, the most general form of a Markovian, completely positive, trace-preserving evolution for a density operator. This choice was not grounded in physical assumptions about the brain, but rather in the structural properties of the formalism, which ensure consistency of probability assignments and facilitate modeling of irreversible processes in open, context-dependent systems.

Cognitive states are represented by density operators on a Hilbert space of cognitive observables. Their evolution is governed by GKSL equation:
\begin{equation}
\label{GKSL}
\frac{d\rho}{dt}
=
- i [H, \rho]
+
\sum_k
\left(
L_k \rho L_k^\dagger
-
\frac{1}{2}
\{L_k^\dagger L_k, \rho\}
\right),
\end{equation}
where $H$ encodes internal cognitive tendencies and the jump operators $L_k$ describe interaction of a system $S$ with 
the environment ${\cal E}.$ Here ${\cal E}.$ is understood as psycho-mental environment of an individual or generally as the information environment. With this general interpretation, this equation is used in QCDM, biology (including genetics, epigenetics, evolution theory), game theory \cite{Asano1}-\cite{QL3}, \cite{BioBas,Open_QLM} 
 as well as in QLM in economics, finance, political and social sciences \cite{Open_QLM,KhrennikovaHaven2016}. In modeling bridging of classical dynamics in neuronal networks of the brain with QCDM \cite{Iriki}, a system $S$ is a network realizing one special psychic function and its environment ${\cal E}$ consists of signals coming from other networks interacting with $S.$ In such models $S$ is composed not of individual neurons, but of neuronal circuits. However, we shall not consider such bridging problem in this paper.    

Denote the superoperator in right hand side of the GKSL equation (Liouvillian) by $  {\cal L},$ it acts in the space of linear operators 
$L({\cal H}),$ where complex Hilbert space${\cal H}$ is system's state space (in the infinite dimensional case one considers the space of bounded linear operators). Then the GKSL equation can be written as 
\begin{equation}
\label{GKSL1}
\frac{d\rho}{dt} =  {\cal L}  \rho(t).
\end{equation}    
The operator $  {\cal L}$ is the infinitesimal generator of the semigroup consisting of superoperators defined by 
\begin{equation}
\label{GKSL2}
P_t   \rho = e^{   {\cal L} t}   \rho
\end{equation}    
acting in $L({\cal H})$ and  determining quantum Markov dynamics. We understood well that processing of information by brain (as well as e.g. by financial market) can be  described by the quantum Markovian dynamics only approximately. 

\subsection{Cognitive Interpretation of Dynamical Terms}

The Hamiltonian term, mapping $  \rho \to - i [  H,   \rho],$  captures internal cognitive dynamics such as oscillations between competing alternatives, conflict monitoring, and internal deliberation.\footnote{In cognitive science, the term deliberation usually refers to the process of weighing alternatives prior to making a decision. Historically, deliberation has often been understood as an internal process, associated with internal reasoning, reflective evaluation, and conscious comparison of options.
However, in modern decision science the notion of deliberation is interpreted more broadly. It may also include external information acquisition, interaction with the environment, sequential accumulation of evidence, and dialogue with media or social inputs. In this wider sense, deliberation refers to the overall dynamics of decision formation and is not restricted to purely internal cognitive activity. In the present paper, we adopt this broader interpretation and use the term {\it deliberation} to denote the dynamical process of preference formation under both internal and external influences.} The jump operators capture environmental influences  - informational context, memory effects (especially in the case of GKSL dynamics with time dependent operators), social pressure, noise, and other external impacts  - which induce decoherence and probabilistic reduction of cognitive superpositions.
In bridging of classical neuronal dynamics in the brain with QCDM \cite{Iriki}, these terms describe interaction of a neuronal network $S$ with other Neuronal networks. We know that the internal dynamical interpretation of the Hamiltonian term is too straight forward, since some components of the Hamiltonian operator $  H$ can also reflect environment's signaling. We use this interpretation just for simplicity.   

The paper's use of the GKSL equation to model a cognitive agent as an ``open system'' interacting with an ``informational environment'' is a direct application of nonequilibrium statistical mechanics (in physics cf. \cite{Zwanzig}).

\section{Decision making through decoherence}

A central insight of the GKSL approach is that decision making, known as {\it decision making through decoherence}, can be modeled as convergence toward a steady state $\rho_{\rm ss}$  -- a decision state, satisfying
\begin{equation}
\label{SS}
\rho_{\rm ss}= \lim_{t\to \infty} \rho(t).
\end{equation}
In this interpretation, a decision is not an instantaneous event but the outcome of a dynamical process in which, starting with the initial belief  state $\rho_0,$  the cognitive state $\rho(t)$ asymptotically stabilizes to the decision state
$\rho_{\rm ss}= \rho_{\rm ss}(\rho_0)$ under the influence of both internal processes and environmental context.
We illustrated this mechanism primarily in the context of the Prisoner's Dilemma and other strategic games in our earlier works \cite{Asano2,BioBas}, see also section \ref{Prisoner}. We remark that different initial belief states may generate the asymptotic convergence to different decision states; moreover, not every believe state $\rho_0$ generates asymptotic  convergence to some decision state $\rho_{\rm ss}.$  The latter may occur when the Liouvillian $\mathcal{L}$ has nonzero
purely imaginary eigenvalues whose corresponding oscillatory
modes are excited by $\rho_0$.

Each decision state $ \rho_{\rm ss}$ is a stationary point of the GKSL dynamics. The stationary points can be found from the equation 
\begin{equation}
\label{SS1}
{\cal L} ( \rho)= 0.
\end{equation}
Solutions of this equation are all potential candidates for decision states. 
We remark that {\it any GKSL dynamics on a finite-dimensional Hilbert space
admits at least one steady state}.

We remark that generally approaching a decision state $ \rho_{\rm ss}$ is characterized by fluctuations whose amplitudes are damped in the process of stabilization, but never disappear completely. 
Of course, $\lim_{t \to \infty}$ is the mathematical abstraction and a real decision maker ``stops the process of decision generation'' at times sharply distributed around the relaxation time of the GKSL dynamics.     

In decision-making models, we introduce {\it decision operator} (Hermitian) $A$ which is diagonal in some {\it decision basis} $(|e_n\rangle)$ and decision outcomes are labeled by its eigenvalues $(a_n).$ For simplicity, we can consider a decision operator with non-degenerate spectrum. Our further analysis of decision making through decoherence is based on the following theorem (see, e.g.,  
 \cite{Frigerio,Kossakowski,Spohn} for details).  

\begin{theorem}[Existence of a Diagonal Steady State]
\label{T1}
Consider a finite-dimensional Hilbert space $\mathcal{H}$ with an orthonormal basis $\{|e_n\rangle\}_{n=1}^N$.  
Let the system evolve under the GKSL equation:
\[
\frac{d\rho}{dt} = \mathcal{L}[\rho] = -i[H,\rho] + \mathcal{D}(\rho),
\]
with generator's components satisfying the following restrictions:    
\begin{equation}
\label{com}
 [H, |e_n\rangle\langle e_n|]= 0 \; \mbox{for all}\; n;
 \end{equation}
decoherence operator $ \mathcal{D}$ is composed of the transition  and dephasing operators:
\begin{equation}
\label{tr}
 L_{nm} = \sqrt{\gamma_{nm}}|e_n\rangle\langle e_m|, \;  n \neq m,
\end{equation}
\begin{equation}
\label{dph}
L_{nn} =  \sqrt{\gamma_n} |e_n\rangle\langle e_n|,
\end{equation}
with non-negative rate coefficients, $\gamma_{nm} \ge 0, \gamma_n \ge 0,$  such that Sufficient Damping Condition holds: 

For every pair of indices $n \neq m$, the total decoherence rate $\Gamma_{nm}$ is strictly positive:
    \[
    \Gamma_{nm} = \frac{1}{2} \left( \sum_{k \neq n} \gamma_{kn} + \sum_{k \neq m} \gamma_{km} \right) + \frac{1}{2}(\gamma_n + \gamma_m) > 0.
    \]

\textbf{Conclusion:}  
For any initial density matrix $\rho(0)$, the evolution of the populations $p_n(t) = \langle e_n | \rho(t) | e_n \rangle$ is governed by the dynamical Pauli equation (cf. with physics and chemistry \cite{vanKampen}):
\[
\frac{d p_n(t)}{dt} = \sum_{m \neq n} \big( \gamma_{nm} p_m(t) - \gamma_{mn} p_n(t) \big).
\]
As $t \to \infty$, the system converges to a steady state $\rho_{\rm ss}$ that is \textbf{diagonal in the $\{|e_n\rangle\}$ basis}:
\[
\rho_{\rm ss} = \sum_n p_n^{\rm ss} |e_n\rangle\langle e_n|.
\]
The populations $p_n^{\rm ss}$ satisfy the stationary Pauli equation:
\[
\sum_{m \neq n} \big( \gamma_{nm} p_m^{\rm ss} - \gamma_{mn} p_n^{\rm ss} \big) = 0 \quad \forall n.
\]
\end{theorem}

We note that commutativity condition (\ref{com}) is equivalent to diagonalization of Hamiltonian in basis $(|e_n\rangle):$
\begin{equation}
\label{com1}
 H = \sum_n E_n |e_n\rangle\langle e_n|.
 \end{equation}

We remark that the diagonal operators $L_{nn}$ (pure dephasing) increase the decoherence rate $\Gamma_{nm}$ but do not affect the steady-state populations $p_n^{\rm ss}$. We also note that the dynamics of the off-diagonal elements of the state operator satisfy the dephasing equation:
\[
\frac{d\rho_{nm}}{dt} = \left[ -i(E_n - E_m) - \Gamma_{nm} \right] \rho_{nm}, \quad (n \neq m).
\]
Since $\Gamma_{nm} > 0$, $\lim_{t \to \infty} \rho_{nm}(t) = 0$.
		
\section{Classical-like decision making: Passive Hamiltonian} 

In particular, we seek a dynamics where the steady state $\rho_{\rm ss}$ is diagonal in a specific decision basis $\{|e_n\rangle\}$, representing the complete resolution of cognitive interference into discrete probabilities. From a cognitive perspective, this construction models the transition from a state of mental superposition---where multiple mutually exclusive options are considered simultaneously---to a state of definite commitment. This is the classical-like regime  of decision making represented within quantum formalism. We utilize the output of Theorem 1. 

The dissipative operator $\mathcal{D}$ is constructed of 
\begin{itemize}
    \item \textbf{Dephasing Operators}  \cite{Schlosshauer}: To eliminate off-diagonal elements without altering the population of the states, we define $L_{nn} = \sqrt{\gamma_n} |e_n\rangle \langle e_n|$. This represents the resolution of cognitive dissonance through environmental or internal decoherence.
    \item \textbf{Transition Operators} \cite{Davies,Alicki}: To redistribute populations toward a target distribution, we define $L_{mn} = \sqrt{\gamma_{mn}} |e_m\rangle \langle e_n|, n \not=m.$ These operators describe the active shift in preference from state $|e_n\rangle$ to $|e_m\rangle$.
\end{itemize}		
		
The full dissipative part of the GKSL generator, $\mathcal{D}$, is constructed as the sum of two distinct cognitive processes: dephasing and transition: $\mathcal{D} =  \mathcal{D}_{deph} + \mathcal{D}_{trans}.$
Psychologically, these components represent distinct functions: dephasing is the process of categorization and the thinning of the mental field, where the ambiguity of holding conflicting beliefs is resolved. Transitions represent the process of active deliberation or preference shifting, where the subject weighs the relative merits of outcomes to favor the most attractive strategy.

The dephasing component, representing the environment-induced loss of cognitive coherence, is defined by:
\begin{equation}
\mathcal{D}_{deph}(\rho) = \sum_n \gamma_n \left( |e_n\rangle\langle e_n| \rho |e_n\rangle\langle e_n| - \frac{1}{2} \{|e_n\rangle\langle e_n|, \rho\} \right),
\end{equation}
where $\gamma_n$ is the decoherence rate. These operators specifically suppress the off-diagonal elements $\rho_{jk}$ ($j \neq k$), ensuring the state becomes a classical mixture in the strategy basis $\{|e_n\rangle\}$.

The transition component is defined by:
\begin{equation}
\mathcal{D}_{trans}(\rho) = \sum_{m \neq n} \gamma_{mn} \left( |e_m\rangle\langle e_n| \rho |e_n\rangle\langle e_m| - \frac{1}{2} \{|e_n\rangle\langle e_n|, \rho\} \right)
\end{equation}
where the coefficients $\gamma_{mn}$ are transition rates

Under the condition (\ref{com}) for the decision basis vectors, the internal cognitive logic (described by Hamiltonian  $H)$ cannot drive population shifts.  Under the Sufficient Decoherence Condition, the resulting decision state $\rho_{\rm ss}$ is a purely classical mixture:
\begin{equation}
\rho_{\rm ss} = \sum_n p_n |e_n\rangle \langle e_n| .
\end{equation}

\subsection{Uniqueness vs. non-uniqueness of decision state}
\label{un}

\begin{theorem}[Uniqueness of the Steady state]
\label{T1a}
Let the condition of Theorem \ref{T1} hold and additionally let the transition rates $\gamma_{nm}$ be irreducible (the state graph is strongly connected),  Then steady state $\rho_{\rm ss}$ is \textbf{unique}. 
\end{theorem}

In this case the system is ergodic:  for any initial state $\rho(0)$, the dynamics converge to this unique diagonal state:
 \[ \lim_{t \to \infty} \rho(t) = \rho_{\rm ss}. \]

{\bf \begin{remark} (Uniqueness beyond irreducibility). Irreducibility is a sufficient but not necessary condition for the uniqueness of the stationary state. For a finite-state continuous-time Markov chain, the number of linearly independent stationary probability distributions equals the number of closed communicating classes. Consequently, the stationary distribution is unique if and only if the transition graph contains exactly one terminal strongly connected component. In this case, the stationary distribution is globally attractive, even when the full transition graph is reducible.

If the graph contains several terminal communicating classes, the stationary distribution is not unique. The limiting distribution then depends on the initial condition through the absorption probabilities associated with these classes.
\end{remark}}

We recall that he transition rates are said to be irreducible if, for every pair of indices $(i, j)$, there exists a finite sequence of intermediate indices $k_1, k_2, \dots, k_r$ such that the product of the corresponding transition rates is strictly positive:
    \begin{equation}
    \gamma_{i k_1} \cdot \gamma_{k_1 k_2} \cdot \dots \cdot \gamma_{k_r j} > 0.
    \end{equation}
Thus,  it is possible to transition from any basis state $|e_j\rangle$ to any other basis state $|e_i\rangle$ through a finite chain of transitions with non-zero rates. 

This definition can be formulated in the  graph theoretic framework as strong connectivity. The state graph is a directed graph $G = (V, E)$ defined as follows:
    \begin{itemize}
        \item \textbf{Vertices ($V$):} The set of orthonormal basis states $\{|e_1\rangle, |e_2\rangle, \dots, |e_N\rangle\}$.
        \item \textbf{Edges ($E$):} A directed edge exists from state $|e_m\rangle$ to state $|e_n\rangle$ if and only if the transition rate is strictly positive ($\gamma_{nm} > 0$).
    \end{itemize}
    The graph $G$ is strongly connected if, for every pair of vertices $u, v \in V$, there exists a directed path from $u$ to $v$ and a directed path from $v$ to $u$.

\paragraph{Cognitive Interpretation of Graph Connectivity.}
The transition graph describes the possible trajectories of
thought, represented by the evolution of the density matrix
$\rho(t)$ toward a stationary state $\rho_{\mathrm{ss}}$.
Strong connectivity means that every decision state is
reachable from every other through transitions with
non-zero rates, representing cognitive flexibility and
open-mindedness.

\paragraph{The Decision State as a Universal Attractor.}
By Theorem~2, strong connectivity guarantees a unique,
globally attractive stationary state. In QCDM, this means
that the final decision distribution is independent of the
initial mental state $\rho(0)$, representing convergence
under the influence of the informational environment and
internal dynamics.

\paragraph{Cognitive Traps and Dogmatism (Reducible Case).}
If the transition graph is reducible, it may contain several
terminal communicating classes, from which no transitions
lead to other classes. If there is exactly one terminal
class, the stationary distribution remains unique and
globally attractive. If there are several, it is not unique,
and the long-term distribution depends on the initial state
through the probabilities of entering each class. In QCDM,
these classes may represent persistent beliefs or cognitive
silos, with the initial conditions influencing the eventual
decision regime.

\paragraph{Transient Thoughts vs. Persistent Beliefs.}
States outside the terminal classes may be transient:
they can be visited during deliberation but are eventually
left with probability one. Long-term dynamics are governed
by the terminal classes, which represent persistent beliefs
or stable decision regimes.

\subsection{Detailed Balance Condition}
\label{DBC}

In cognitive applications, it is natural to appeal to the Detailed Balance Condition as a special realization the Sufficient Decoherence Condition. We note that the Detailed Balance Condition is stronger than  the Sufficient Decoherence Condition and Theorem \ref{T1a} can be applied.  By the Detailed Balance Condition the probability flow from state $n$ to $m$ equals the flow from $m$ to $n$:
\begin{equation}
p_n \gamma_{mn} = p_m \gamma_{nm}
\end{equation}
For a system driven by utilities $u_n$ and $u_m$, the transition rates are typically defined as 
\begin{equation}
\label{u}
\gamma_{mn} = \Gamma e^{\frac{\beta}{2}(u_m - u_n)}.
\end{equation}
Substituting this into the balance equation yields the ratio:
\begin{equation}
\frac{p_m}{p_n} = \frac{\gamma_{mn}}{\gamma_{nm}} = \frac{e^{\frac{\beta}{2}(u_m - u_n)}}{e^{\frac{\beta}{2}(u_n - u_m)}} = e^{\beta(u_m - u_n)}
\end{equation}
The unique normalized solution to this system is the Gibbs distribution:
\begin{equation}
\label{GD}
p_n = \frac{e^{\beta u_n}}{\sum_k e^{\beta u_k}},
\end{equation}
where $\beta$ acts as the intensity of choice, which is mathematically identical to the inverse temperature in thermodynamics
This demonstrates that in the commuting limit (in conjunction with the Detailed Balance Condition), the GKSL equation generates the classical {\it Quantal Response Equilibrium} - Gibbs distribution (\ref{GD}). Such GKSL dynamics can be treated as classical.

Dynamics constrained by conditions of Theorem \ref{T1} can be certified as classical-like. It generates a decision state that is diagonal in the decision basis (generally it is non-unique and need not coincide with Quantal Response Equilibrium. 

\section{Quantum decision making: Active Hamiltonian}
\label{AH}

If the Hamiltonian contains off-diagonal elements,
\begin{equation}
		\label{NCH1}
    [H, |e_n\rangle \langle e_n|] \not= 0,
    \end{equation}
it acts as a coherent drive that rotates population between different states. Mathematically this situation is described by the following theorem:

\begin{theorem}[Non-Existence of Diagonal Steady States]
\label{T2}
Let $\mathcal{L}[\rho_{\rm ss}] = 0$ and let the dissipator have the transition and dephasing
structure specified in Theorem \ref{T1}. If there exists a pair $(n, m)$ such that:
\begin{enumerate}
    \item \textbf{Non-zero Hamiltonian Coupling:} $\langle e_n | H | e_m \rangle \neq 0$
    \item \textbf{Population Gradient:} $\langle e_n | \rho_{\rm ss} | e_n \rangle \neq \langle e_m | \rho_{\rm ss} | e_m \rangle$
\end{enumerate}
Then $\rho_{\rm ss}$ \textbf{cannot} be diagonal in the decision basis $\{|e_n\rangle\}$.
\end{theorem}

\begin{proof}
Assume, for contradiction, that the steady state $\rho_{\rm ss}$ is diagonal in the basis $\{|e_n\rangle\}$. We write $\rho_{\rm ss} = \sum_k p_k |e_k\rangle\langle e_k|$, where $p_k = \langle e_k | \rho_{\rm ss} | e_k \rangle$. For $\rho_{\rm ss}$ to be a steady state, the matrix elements of the Lindbladian must vanish for all pairs $(n, m)$. For $n \neq m$:
\begin{equation}
\langle e_n | \mathcal{L}[\rho_{\rm ss}] | e_m \rangle = -i \langle e_n | [H, \rho_{\rm ss}] | e_m \rangle + \langle e_n | \mathcal{D}[\rho_{\rm ss}] | e_m \rangle = 0
\end{equation}

1. Dissipative Term: Under the Sufficient Damping Condition, the action of $\mathcal{D}$ on any matrix $\rho$ results in $\langle e_n | \mathcal{D}[\rho] | e_m \rangle = -\Gamma_{nm} \rho_{nm}$. Since we assumed $\rho_{\rm ss}$ is diagonal, $\rho_{nm} = 0$ for $n \neq m$. Thus:
\begin{equation}
\langle e_n | \mathcal{D}[\rho_{\rm ss}] | e_m \rangle = 0
\end{equation}
2. Hamiltonian Term: The stationarity condition now requires the commutator term to vanish independently:
\begin{equation}
\langle e_n | [H, \rho_{\rm ss}] | e_m \rangle = \sum_k (H_{nk} \rho_{km} - \rho_{nk} H_{km}) = 0
\end{equation}
Substituting the diagonal form $\rho_{ij} = p_i \delta_{ij}$, the sum reduces to:
\begin{equation}
H_{nm} p_m - p_n H_{nm} = H_{nm}(p_m - p_n) = 0
\end{equation}
3. Contradiction: By the theorem's hypothesis, there exists a pair $(n, m)$ such that $H_{nm} \neq 0$ and $p_n \neq p_m$. This implies $H_{nm}(p_m - p_n) \neq 0$, which contradicts the requirement that the Hamiltonian term must be zero. 

Therefore, the assumption that $\rho_{\rm ss}$ is diagonal must be false. The steady state must satisfy $\rho_{nm} \neq 0$ for at least one pair $(n, m)$.
\end{proof}

{\it Cognitive Meaning:} This theorem proves that internal deliberation
($H$) acting up to asymmetric preferences $(p_n \not= p_m)$ forces the agent into a
non-classical state of persistent coherence. In this regime:
\begin{itemize}
    \item The Hamiltonian constantly generates coherences ($\rho_{jk} \neq 0$) that the dephasing term $\mathcal{D}_{deph}$ attempts to destroy.
    \item The stationary state $\rho_{\rm ss}$ is not purely diagonal in the decision basis (even under Sufficient Decoherence Condition). It retains non-vanishing off-diagonal elements, representing a state of ``active deliberation'' or persistent quantum-like superposition.
At least, it is the cognitive state where the property of representational simultaneity is not completely vanished.
Its eigenstates will not correspond to the events (the responses) observed in the nature, since each of them is the superposition of $\{|e_n\rangle\}$, generally. However, the agent may experience these cognitive states in its mind: Such the experience certainly exists in the epistemological sense, not in the natural scientific one. We have to distinguish the epistemological reality from the natural scientific reality.   
    \item Crucially, the presence of ``active Hamiltonian'' $H$ can shift the diagonal populations $p_n$ away from the values predicted by the transition rates alone (coherence-induced population shift, the well known result in quantum thermodynamics). The internal logic represented by $H$ competes with the environmental pressure of the jump operators, leading to a driven-dissipative steady state that does not align with classical expectations.
\item Purely imaginary eigenvalues of the generator lead to non-decaying
oscillations of the cognitive state. In decision-making terms, this
represents stable indecision: the agent cyclically shifts preference
between alternatives without converging to a final probability
distribution. The mind remains in a self-sustained deliberative loop,
with competing options repeatedly reactivated. Such behavior may arise
in certain pathological cognitive conditions, where thought processes
become trapped in repetitive cycles, for example in obsessive rumination,
manic switching of intentions, or other disorders characterized by
persistent cognitive instability.
\end{itemize}

\begin{proposition}[Structure of a coherence-induced stationary state]
\label{prop:stationary-coherence}
Assume that the dissipator has the transition and dephasing
structure specified in Theorem \ref{T1}, so that its action on the
off-diagonal matrix elements is
\[
  \mathcal{D}(\rho)_{nm}=-\Gamma_{nm}\rho_{nm}.
  \qquad n\ne m,
\]
 Write the Hamiltonian as
\[
  H=\sum_n E_n |e_n\rangle\langle e_n|+V,
\]
where \(V_{nn}=0\) and \(V_{nm}=H_{nm}\) for \(n\ne m\).

Let \(\rho_{\mathrm{ss}}\) be a stationary state and denote
its populations by
\[
  p_n=\langle e_n|\rho_{\mathrm{ss}}|e_n\rangle.
\]
Then, for every \(n\ne m\), its off-diagonal elements satisfy
\begin{align}
  &\bigl(\Gamma_{nm}+i(E_n-E_m)\bigr)
       (\rho_{\mathrm{ss}})_{nm}
  \nonumber\\
  &\quad =
  iH_{nm}(p_n-p_m)
  -i\sum_{k\ne n,m}
  \left(
    H_{nk}(\rho_{\mathrm{ss}})_{km}
    -(\rho_{\mathrm{ss}})_{nk}H_{km}
  \right).
  \label{eq:stationary-coherence-general}
\end{align}

Consequently, the stationary state has the decomposition
\[
  \rho_{\mathrm{ss}}
  =
  \sum_n p_n |e_n\rangle\langle e_n|
  +
  \sum_{n\ne m}c_{nm}|e_n\rangle\langle e_m|,
\]
where \(c_{nm}=(\rho^{ss})_{nm}\) satisfy the coupled off-diagonal stationary equations 
together with the diagonal stationary equations.

In particular, if the Hamiltonian couples only the pair of
alternatives \(n\) and \(m\), so that the sum in
 (\ref{eq:stationary-coherence-general}) vanishes, then
\begin{equation}
  c_{nm}
  =
  \frac{iH_{nm}(p_n-p_m)}
       {\Gamma_{nm}+i(E_n-E_m)},
  \label{eq:stationary-coherence-pair}
\end{equation}
\end{proposition}
provided that \(\Gamma_{nm}+i(E_n-E_m)\neq 0\).

\begin{proof}
The stationary equation is
\[
  -i[H,\rho_{\mathrm{ss}}]
  +\mathcal{D}(\rho_{\mathrm{ss}})=0.
\]
For \(n\ne m\), its matrix element in the decision basis gives
\[
  0=-i[H,\rho_{\mathrm{ss}}]_{nm}
    -\Gamma_{nm}(\rho_{\mathrm{ss}})_{nm}.
\]
Using the decomposition of \(H\), we obtain
\begin{align*}
 [H,\rho_{\mathrm{ss}}]_{nm}
  &=(E_n-E_m)(\rho_{\mathrm{ss}})_{nm}
    +H_{nm}(p_m-p_n)\\
  &\quad+
    \sum_{k\ne n,m}
    \left(
      H_{nk}(\rho_{\mathrm{ss}})_{km}
      -(\rho_{\mathrm{ss}})_{nk}H_{km}
    \right).
\end{align*}
Substitution into the stationary equation and rearrangement
give  (\ref{eq:stationary-coherence-general}). If the Hamiltonian
couples only the pair \(n,m\), the remaining sum vanishes,
yielding  (\ref{eq:stationary-coherence-pair}).
\end{proof}

In physics, quantum information theory,  this area of research (the GKSL dynamics with Active Hamiltonians) is known as {\it Dissipative State Engineering} \cite{Alicki,Diehl}.

\subsubsection*{The Necessity of Non-Commutation: Agency and Conflict Monitoring}

For a model to exhibit a Quantum Escape from a classical equilibrium, the non-commutation condition $[H, |e_n\rangle \langle e_n|] \neq 0$ is mandatory. From a cognitive perspective, this non-zero commutator represents the mathematical signature of cognitive agency or subjective bias, where the internal cognitive drive actively resists the dissipative pressure of the informational environment.

Theorem \ref{T2} remains valid when the transition rates satisfy the Detailed Balance Condition. In particular, whenever the active Hamiltonian couples two alternatives with unequal stationary populations, the corresponding stationary state cannot be diagonal in the decision basis. Thus, detailed balance of the dissipative transition rates does not by itself guarantee a diagonal Gibbs stationary state in the presence of an active Hamiltonian.

Even under the Detailed Balance Condition, an Active Hamiltonian can act as a {\it conflict monitor} \cite{Botvinick2001, Ridderinkhof2004}. In this regime, $H$ continuously drives the generation of off-diagonal coherences, counteracting their dissipative suppression and thereby sustaining quantum-like coherences $\rho_{jk}$ in the stationary state. This prevents the mental state from settling into a purely classical mixture. This mechanism provides a dynamical explanation for key phenomena in QCDM:

{\it Persistent Deliberation}: The steady state $\rho_{\rm ss}$ remains ``warm,'' maintaining a state of active tension between competing alternatives rather than a cold, singular choice.

{\it Quantum Escape in Strategic Games:} The internal logic of the decision-maker can stabilize outcomes that are classically unstable. In the \textit{Prisoner's Dilemma} (see section \ref{Prisoner} for details), for instance, while dissipative transition operators $\mathcal{D}_{trans}$ drive the system toward the Nash equilibrium of mutual defection, a non-commuting Hamiltonian can 
 shield the cooperative superposition. This allows the system to escape the classical utility trap and stabilize a decision state with a significant probability of mutual cooperation, reflecting the persistent logical intuitions observed in human subjects \cite{DeNeys2012}.
In this light, the magnitude of the commutator $[H, |e_n\rangle \langle e_n|]$ serves as a formal measure of the individual's autonomy from the external reward structure.

\section{Cognitive beats: a distinguishing property of decision by coherence model}

\subsection{Multifrequency beats}

Within the QCDM framework, a two-frequency beat (cognitive beat) arises from the interference of two oscillatory components of the probability trajectory with nearby frequencies. Equal amplitudes are not required for the appearance of beats, although they produce the most pronounced and symmetric modulation pattern. This phenomenon serves as a signature of the coupling between the internal Hamiltonian dynamics and the dissipative GKSL environment, and represents a spectral feature of the Liouvillian operator.

\subsubsection*{General two--frequency superposition}

Consider two oscillating contributions to the probability trajectory
\begin{equation}
p(t) = A_1 \cos(\omega_1 t) + A_2 \cos(\omega_2 t),
\end{equation}
where the amplitudes $A_1$ and $A_2$ may be different and the frequencies are assumed to be close,
$|\omega_1-\omega_2| \ll \omega_{1,2}.$ Introducing the average and difference frequencies
$$
\omega_{\mathrm{avg}} = \frac{\omega_1+\omega_2}{2}, \Delta \omega = \omega_1-\omega_2 ,
$$
the trajectory can be rewritten as
\begin{equation}
p(t) = A_1 \cos\left(\omega_{\mathrm{avg}} t + \frac{\Delta\omega}{2}t\right)
+ A_2 \cos\left(\omega_{\mathrm{avg}} t - \frac{\Delta\omega}{2}t\right).
\end{equation}
This representation shows that the signal consists of a fast oscillation at frequency $\omega_{\mathrm{avg}}$ whose amplitude varies slowly on the time-scale $1/|\Delta\omega|$. The slow modulation corresponds to the beat phenomenon. The beat frequency is determined by the difference of the two oscillation frequencies,
$\omega_{\mathrm{beat}} = |\omega_1-\omega_2|.$

The amplitude envelope of the two-frequency signal can be expressed
explicitly as
\begin{equation}
R(t)=\sqrt{A_1^2+A_2^2+
2A_1A_2\cos(\Delta\omega t)},
\end{equation}
where $\Delta\omega=\omega_1-\omega_2$. This envelope describes the
slowly varying amplitude of the combined oscillations and satisfies
\begin{equation}
|A_1-A_2|\leq R(t)\leq A_1+A_2.
\end{equation}
Thus, when the amplitudes are unequal, the envelope does not reach
zero. In the symmetric case $A_1=A_2=A$, it reduces to
\begin{equation}
R(t)=2A\left|\cos\left(\frac{\Delta\omega}{2}t\right)\right|,
\end{equation}
and periodically vanishes, producing the characteristic nodes of a
symmetric beat pattern.

\subsubsection*{Symmetric beat: equal amplitude case}

A particularly transparent form is obtained when the amplitudes coincide,
 $A_1 = A_2 = A.$ In this symmetric case the signal reduces to
\begin{equation}
p(t) = 2A
\cos\left( \frac{\omega_1-\omega_2}{2} t \right)
\cos\left( \frac{\omega_1+\omega_2}{2} t \right).
\end{equation}
In this equal-amplitude limit the envelope periodically reaches zero, producing maximal constructive and destructive interference. This corresponds to the canonical beat pattern with clearly visible nodes and antinodes, see Figure \ref{Basic}.Montroll

\newpage

\begin{figure}
\begin{center} 
\includegraphics[width=1\textwidth, scale=0.5]{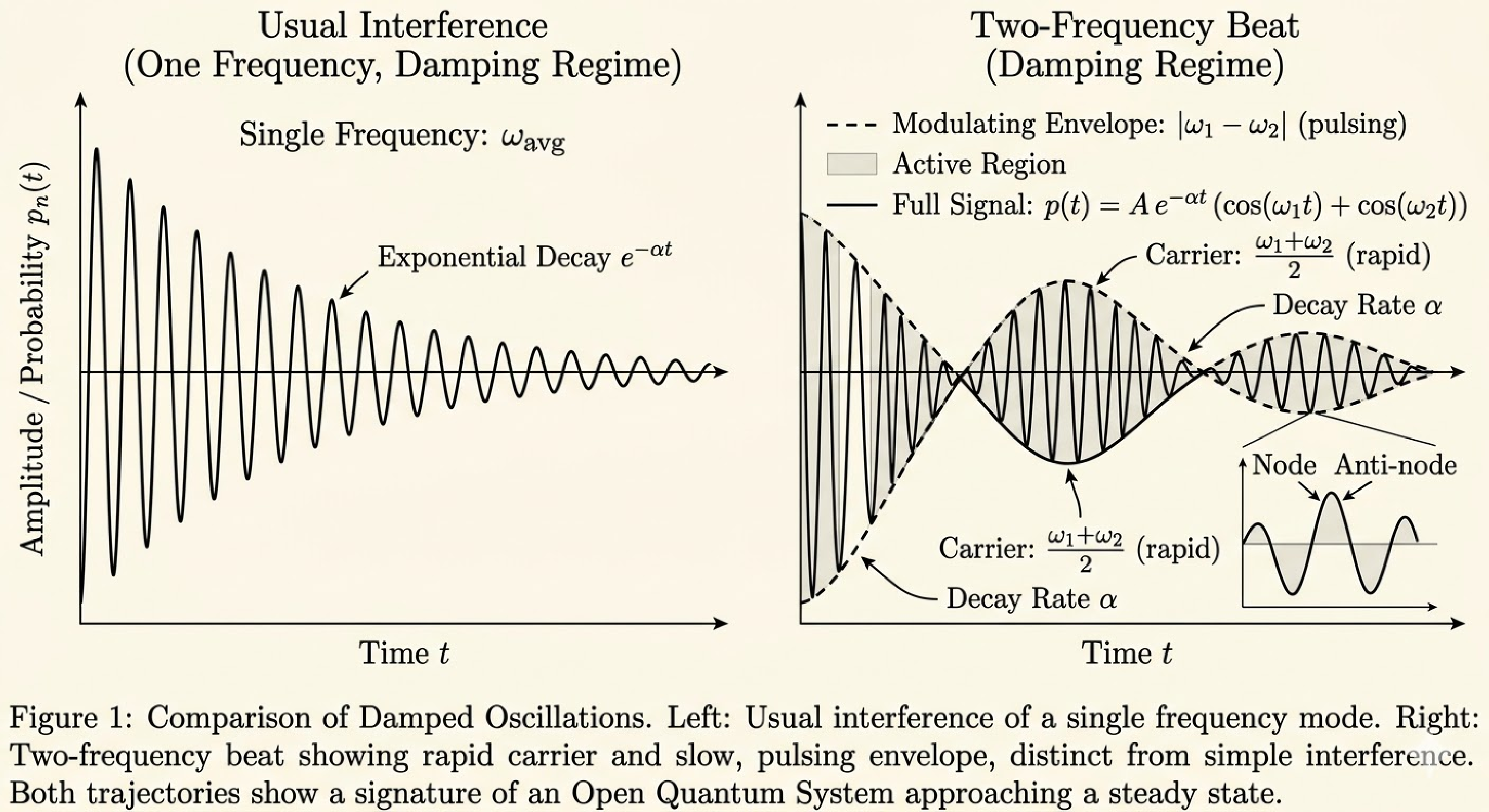}\hspace{4mm} 
\end{center}
\label{Basic}
\caption{Cognitive beats.}
\end{figure}

\newpage

\subsubsection*{Interpretation in GKSL / QCDM dynamics}

The appearance of beats indicates the presence of at least two oscillatory modes of the Liouvillian dynamics with nearby imaginary parts of the eigenvalues. The frequencies $\omega_1$ and $\omega_2$ are determined by the Liouvillian spectrum, while the amplitudes $A_1$ and $A_2$ depend on the projection of the initial state onto the corresponding eigenmodes.

Thus the conditions for beats are: (a) close Liouvillian frequencies $\omega_1 \approx \omega_2$ (necessary condition),
 (b) nonzero contributions of both modes to the observed probability trajectory, determined jointly by the initial state and the projection onto the corresponding observable. We also highlight again that
 equal amplitudes produce maximal symmetric beats. The pulsing behaviour of \(p(t)\) therefore serves as a diagnostic of multimode dissipative dynamics and indicates the presence of multiple oscillatory channels in the QCDM evolution.

\subsubsection*{General multi-frequency formulation}

In a complex dissipative cognitive system, the probability trajectory is typically composed of multiple oscillatory modes originating from the Liouvillian spectrum. The oscillatory part of the probability trajectory can be represented in the form:
\begin{equation}
\label{MF}
p_n(t) = p_n^{ss} + \sum_{k} c_k e^{-\alpha_k t} \cos(\omega_k t + \phi_k)
\end{equation}
where $p_n^{ss}$ is the stationary probability, $\alpha_k$ are decay rates determined by the real parts of Liouvillian eigenvalues,  
$\omega_k$ are oscillation frequencies determined by the imaginary parts,
 $c_k$ are amplitudes fixed by the initial state,
 $\phi_k$ are phase shifts.
A \textit{multifrequency beat} occurs whenever the set of frequencies ${\omega_k}$ contains at least two values that are relatively close to each other, $
\omega_i \approx \omega_j . $

In this situation the interference between these nearby modes produces slow amplitude modulation on top of faster oscillatory dynamics. If several clusters of close frequencies are present, the trajectory may exhibit nested beats, hierarchical modulation, or quasi-periodic pulsing patterns. Such behaviour is a direct signature of a structured Liouvillian spectrum and indicates the presence of multiple competing dynamical channels.

In the special case of two dominant nearby modes, the general multi-frequency expression reduces to the two-frequency beat discussed above. When more than two nearby frequencies are present, the resulting trajectory can display irregular envelopes, secondary beat frequencies, and long-time modulation patterns.

The following fields provide the mathematical basis for treating beats as a signature of complex dissipative dynamics:
quantum dissipative systems \cite{BreuerPetruccione2002} (analysis of GKSL dynamics and the spectral properties of Liouvillian operators), spectral theory of Liouvillians \cite{Minganti} (discussion how clusters of complex eigenvalues lead to metastable oscillations and beat-like structures),  quantum optics - quantum beats \cite{Scully} (interference of transition frequencies leading to quantum beats in fluorescence signals),  high energy physics \cite{Hiesmayr}.

\subsection{Three mechanisms of cognitive beat phenomena}

Within the QCDM framework, beats arise when at least two
oscillatory Liouvillian modes with close frequencies
contribute to the same probability trajectory. Their
interference produces slow modulation of decision
probabilities. Depending on the origin of these modes,
three mechanisms can be distinguished.

\paragraph{Type I: Hamiltonian (internal deliberation) beats.}

In this case, oscillations originate from the internal
Hamiltonian, while the dissipative part determines their
damping and relaxation. If the Hamiltonian eigenvalues
are $E_m$, the transition frequencies are
$\omega_{mn}=E_m-E_n$. Close transition frequencies
may interfere and produce beats.

Cognitively, this mechanism represents competition
between nearly equivalent internal reasoning pathways.
For example, consider three states representing
indecision, preference for a safe option, and preference
for a risky option. A possible Hamiltonian is
$
H=
\begin{pmatrix}
0 & g_1 & g_2\\
g_1 & \omega_1 & 0\\
g_2 & 0 & \omega_2
\end{pmatrix}.
$
The couplings $g_1$ and $g_2$ connect the undecided
state to the two alternatives. When the corresponding
transition frequencies are close, the probabilities of
choosing the safe or risky option may exhibit slow
modulation. The oscillations originate from internal
deliberation, whereas the environment determines
their damping.

\paragraph{Type II: Dissipative (environment-driven) beats.}

Oscillatory modes may also arise primarily from the
dissipative part of GKSL dynamics, even when the
Hamiltonian is negligible. This requires structured
active dissipative channels beyond ordinary transition
and dephasing operators. Such channels may generate
off-diagonal coherences and contribute to complex
Liouvillian eigenvalues, allowing the environment to
sustain oscillatory dynamics.

A cognitive example is exposure to opposing information
streams. Consider three states: neutrality, a pro
position, and a contra position. Dissipative operators
describe transitions from neutrality toward either
opinion, as well as cross-influences between the two
opinions. When these channels generate oscillatory
Liouvillian modes with close frequencies, their
interference produces beats. Here, the oscillations
reflect the structure of external informational
influences rather than internal deliberation cycles.

\paragraph{Type III: Hybrid beats.}

Hybrid beats result from the interplay between internal
Hamiltonian dynamics and environmental interactions.
Both contribute to the Liouvillian spectrum, whose
eigenvalues may take the form
$
\lambda_k=-\alpha_k\pm i\omega_k(H,\mathcal{L}_{D}),
$
where the frequencies depend jointly on the Hamiltonian
and the dissipative operators. Close frequencies produce
multifrequency beats.

For example, in an investment decision, internal
reasoning may generate oscillations between optimistic
and cautious evaluations, while market news introduces
additional external influences. Their interaction may
produce slow modulations in the probability of buying,
waiting, or remaining undecided.

\paragraph{Interpretation for QCDM.}

The three mechanisms represent internal, environmental,
and combined origins of cognitive beats. In each case,
the interference of oscillatory modes produces slow
modulation of decision probabilities. Although such
behaviour cannot arise in two-state classical Markov
models, it may also occur in classical Markov models
with three or more states. Thus, beats alone do not
establish quantum coherence; however, in two-dimensional
systems, their presence provides a signature of
quantum-like dynamics. Within QCDM, beats therefore
serve as a diagnostic of multimode cognitive dynamics
and the interaction of competing dynamical channels.

\subsection{Identifying the active Hamiltonian}

The emergence of these beats serves as a critical diagnostic tool to distinguish between the Passive and Active Hamiltonian  (in fact, classical and quantum) regimes across different system dimensions $N$. Let $F$ be the maximum number of distinct observable frequencies $\omega > 0$.

{\it The Passive Case (Classical Limit):}
In the Passive Case, the $N \times N$ population matrix is decoupled from the coherences. The dynamics are restricted to the spectrum of an $N$-dimensional transition rate matrix. Because one eigenvalue must be zero (steady state), there are $N-1$ degrees of freedom for dynamics.
\begin{itemize}
    \item If $N=2$, $N-1=1$: The eigenvalue is real. No oscillations are possible ($F=0$).
    \item If $N=3$ or $N=4$, $N-1$ allows for at most one pair of complex conjugate eigenvalues ($F=1$).
\end{itemize}
In this regime, beats are mathematically impossible because a beat requires at least two distinct frequencies ($F \geq 2$).

{\it The Active Case (Quantum-Like Coupling):}
In the Active Case, the full $N^2 \times N^2$ Liouvillian couples populations to off-diagonal coherences.
Excluding the stationary eigenvalue, the full Liouvillian has \(N^2-1\) remaining eigenvalues. Their corresponding modes may contribute to the diagonal elements, depending on the initial state and the observable.
\begin{itemize}
    \item If $N=2$, $N^2-1=3$: Allows for one oscillatory pair ($F=1$).
    \item If $N=4$, $N^2-1=15$: Allows for up to seven oscillatory pairs ($F=7$).
\end{itemize}
The following table summarizes the spectral capacity of both regimes to produce oscillations and beats:
\begin{table}[h]
\centering
\begin{tabular}{@{}llll@{}}
\toprule
\textbf{Dimension ($N$)} & \textbf{Regime} & \textbf{Max Distinct Freq. ($F$)} & \textbf{Certification Result} \\ \midrule
$N=2$ (Qubit) & Passive & 0 & Pure Decay (No Oscillation) \\
 & Active & 1 & Pure Oscillation (No Beats) \\ \addlinespace
$N=3$ (Qutrit) & Passive & 1 & Single Frequency (No Beats) \\
 & Active & 4 & \textbf{Beats Possible} \\ \addlinespace
$N=4$ (PD) & Passive & 1 & Single Frequency (No Beats) \\
 & Active & 7 & \textbf{Complex Beats Possible} \\ \bottomrule
\end{tabular}
\caption{Comparison of spectral complexity between Passive and Active regimes.}
\end{table}

\textbf{Theorem [Visibility of Spectral Complexity].}
Let $\mathcal{L}$ be the GKSL generator in an
$N$-dimensional Hilbert space.
\begin{itemize}
\item \textbf{Passive requirement:} If
$
[H,|e_n\rangle\langle e_n|]=0$
for all $n,$
the population trajectory $p_n(t)$ is restricted to
the spectrum of an $N\times N$ transition-rate matrix.
\item \textbf{Active requirement:} If
$
[H,|e_n\rangle\langle e_n|]\neq 0
$
for at least one $n$, the population dynamics may inherit
the spectral complexity of the full $N^2\times N^2$
Liouvillian. 

For $N\geq3$, the presence of a beat pattern
provides a sufficient certificate of an Active Hamiltonian
within our GKSL framework with the transition and dephasing 
dissipators specified in Theorem \ref{T1}.

\end{itemize}

\section{Equivalence of GKSL dynamics with Passive Hamiltonian and classical Markovian dynamics}

The intuition that Passive Hamiltonian GKSL dynamics is equivalent to classical Markovian dynamics is mathematically rigorous. When the Hamiltonian $H$ is ``passive,'' the quantum-specific interference terms effectively vanish from the evolution of the diagonal elements, reducing the system to a classical Master Equation.

If the Hamiltonian $H$ is passive, then the commutator $[H, |e_n\rangle\langle e_n|] = 0$. Under this condition, the following simplifications occur:
\begin{itemize}
    \item \textbf{Decoupling:} The coherences (off-diagonal elements \(\rho_{jk}\)) do not contribute to the evolution of the populations (diagonal elements \(p_n=\rho_{nn}\)).
\item \textbf{Unitary Vanishing:} The term $-i[H, \rho]$ contributes zero to the evolution of the diagonal elements $p_n$.
\end{itemize}
The resulting rate of change for each probability $p_n$ depends exclusively on the other probabilities $p_m$, 
yielding the Pauli master equation (cf. with physics and chemistry \cite{vanKampen}):
\begin{equation}
    \frac{dp_n}{dt} = \sum_m (W_{nm} p_m - W_{mn} p_n),
\end{equation}
where $W_{nm}$ are the transition rates derived from the Lindblad operators $L_k$. This is identical to the form $\dot{\mathbf{p}} = \mathbf{W}\mathbf{p}$ utilized in classical Markovian dynamics.

The classical equivalence imposes strict constraints on the spectral complexity of the system:
\begin{enumerate}
    \item \textbf{Eigenvalue Constraints:} 
		
		For a classical transition rate matrix $W$ of dimension $N$,
the eigenvalues $\lambda$ are constrained by the properties
of stochastic matrices. For $N=4$, there are four eigenvalues,
one of which is always $0$ (representing the steady state).
Assuming a unique stationary state, the remaining three
eigenvalues have negative real parts.
    \item \textbf{The $N=4$ Limit:} The remaining 3 eigenvalues can only manifest in two configurations:
    \begin{itemize}
        \item Three real negative numbers (resulting in pure exponential decay).
        \item One real negative number and one pair of complex conjugates (resulting in a single frequency oscillation).
    \end{itemize}
\end{enumerate}
Since a classical Markov process on 4 states can support at most one pair of complex conjugate eigenvalues, it is limited to a single frequency $\omega_1$. Because a multifrequency beat requires the interference of at least two distinct frequencies ($\omega_1 \neq \omega_2$), such patterns are mathematically impossible in the Passive/Classical regime, for N=4 (e.g.the Prisoner's Dilemma, section \ref{Prisoner}). 

\paragraph{Detailed balance: monotonic dynamics}

To contrast classical and quantum behavior, it is useful to consider GKSL dynamics with a passive Hamiltonian $H$ and a dissipator ${\cal D}$ satisfying the detailed balance condition (Section~\ref{DBC}). In this case, the population dynamics is given a classical master equation $\dot p = W p$, where the generator $W$ is similar to a symmetric matrix. Consequently, $W$ has only real eigenvalues, and the populations $p_n(t)$ evolve as sums of purely exponential terms without oscillatory components, leading to relaxation toward the stationary state.

In contrast, for an active Hamiltonian, the populations may couple to coherences, so that their evolution is generally not described by a closed classical Markov generator. The effective dynamics are then governed by the full Liouvillian, whose spectrum may contain complex-conjugate eigenvalues. Consequently, the population trajectories may exhibit oscillations or beat patterns.

\section{Contrasting the Busemeyer and Asano-Khrennikov approaches to cognitive modeling}

The role of multifrequency beats serves as the primary differentiator between the unitary model of the Busemeyer et al. \cite{Busemeyer1,Busemeyer2,Pothos3} and the Open-System model (decision making through decoherence) of the Asano-Khrennikov et al. \cite{Asano1}\cite{QL3}. 

\begin{enumerate}
    \item \textbf{Unitary Model:} Dynamics are governed by $U(t) = e^{-iHt}$. For $N=4$, the probability $p_n(t)$ is a sum of frequencies derived from the 4 eigenvalues of $H$. While multifrequency interference is possible, the lack of natural damping often results in sustained oscillations rather than the localized pulses characteristic of cognitive beats.
    \item \textbf{Open-System Model:} Dynamics are governed by the Liouvillian $\mathcal{L}$. The $N \times N$ structure allows for up to numerous distinct frequencies. In this framework, the 'Active Hamiltonian' acts as the engine of deliberation, pumping energy between coherences and populations. This interaction manifests as damped beats—the mathematical signature of a mind moving from uncertainty toward a decision. For a three-option state ($N=3$), these cognitive beats begin to emerge; however, in more complex $N=4$ systems like the Prisoner's Dilemma, the $16 \times 16$ matrix structure allows for a rich spectrum of up to seven distinct frequencies
\end{enumerate}

From the viewpoint of beats, the Asano-Khrennikov framework is more robust for certification. It provides a larger spectral bandwidth ($N^2$ vs $N$), allowing for complex beat patterns that are mathematically impossible in both classical Markovian models and simple unitary models with high symmetry.

\section{The Prisoner's Dilemma}
\label{Prisoner}

The Prisoner’s Dilemma  is the classic ``non-zero-sum'' game where two players must choose between Cooperation ($C$) and Defection ($D$). The main output of the classical expected utility theory is that while mutual cooperation (CC) yields the best collective outcome, individual rationality leads both players to the defect strategy (DD) .

\subsubsection*{The Payoff Matrix}
We consider the standard payoff structure for players Alice (rows) and Bob (columns):

\begin{table}[h]
\centering
\begin{tabular}{lcc}
\toprule
 & Bob: $C$ & Bob: $D$ \\
\midrule
Alice: $C$ & $(R, R)$ & $(S, T)$ \\
Alice: $D$ & $(T, S)$ & $(P, P)$ \\
\bottomrule
\end{tabular}
\caption{Payoff matrix for the Prisoner's Dilemma.}
\end{table}

Here the payoffs follow the hierarchy: $T > R > P > S$ (Temptation $>$ Reward $>$ Punishment $>$ Sucker's payoff).

The Nash Equilibrium occurs at $(D, D)$: regardless of Bob's choice, Alice's payoff is higher if she defects ($T > R$ and $P > S$), 
 the same logic applies to Bob. Despite the fact that $(C, C)$ is Pareto-optimal (both players would be better off), the stable point---where no player can improve their payoff by changing their strategy alone---is mutual defection. This is the output of the classical expected utility theory. However, humans behavior frequently deviates from the Nash equilibrium in the experiments of Prisoner's Dilemma type -- we can rely on several foundational works; we briefly overview some them:
Simon \cite{Simon1955} introduced the concept of Bounded Rationality, providing the historical root for why steady-state Nash equilibria ($t \to \infty$) might not be reached by finite human agents.
In 1979 Kahneman and Tversky \cite{KahnemanTversky1979} wrote the seminal paper on Prospect Theory, which documents how humans deviate from expected utility maximization; this paper is essential for coming discussion on the interplay between classical and quantum-like payoff-driven dynamics (see also Tversky and Kahneman \cite{TverskyKahneman1974} and Shafir and Tversky \cite{ShafirTversky1992}). The paper of Henrich et al. \cite{Henrich2001} provides the empirical ``ground truth'' across different cultures, showing that humans consistently prioritize fairness and cooperation over pure payoff maximization in games like the Prisoner's Dilemma or Ultimatum Game. The paper of Thaler \cite{Thaler1980} explains the psychological mechanisms (like mental accounting) that cause humans to treat potential losses and gains differently, which could be modeled within our QLM as a ``moral'' or ``cognitive'' oscillation.

Representative applications of QLMs to dilemmas such as the Prisoner's Dilemma can be found in 
\cite{Busemeyer0,Busemeyer1,Pothos2}, where interference effects are used to explain violations of the sure-thing principle;
in works \cite{Asano2,Asano4,QL3} the  dynamics of decision making in the Prisoner's Dilemma was modeled with the GKSL equation - decision making through decoherence. 

\subsection{Quantum Dynamics of the Prisoner's Dilemma}

We model the decision process using the basis $\{|CC\rangle, |CD\rangle, |DC\rangle, |DD\rangle\}$. 
The agent assumed in this model is Alice, Bob, or a third entity who tries to predict their cognitive experiences.
The dynamics are governed by the GKSL equation.
To represent the cognitive tension between the Pareto-optimal outcome and the Nash equilibrium, we propose Hamiltonian: 
\begin{equation}
H = \Omega (|CC\rangle\langle DD| + |DD\rangle\langle CC|)
\end{equation}
where $\Omega$ is the frequency of ``cognitive oscillation'' between cooperation and defection. This Hamiltonian doesn't  commute with dephasing operators, so we are in the framework of section \ref{AH}. The steady state $\rho_{\rm ss}$ isn't determined solely by the jump operators. In the full GKSL framework, the stationary condition:
\begin{equation}
\label{Sstate}
-i[H, \rho_{\rm ss}] + \mathcal{D}(\rho_{\rm ss})=0
\end{equation}
shows that $\rho_{\rm ss}$ is the result of a dynamical competition; the stationary state is not simply the thermal or ``rational'' state dictated by the payoffs. When $H$ includes off-diagonal coupling between $|CC\rangle$ and $|DD\rangle$, it acts as a coherent pump. 

Even if the transition rates favor defection, 
$
u_{DD} > u_{CC}
$, the active Hamiltonian can modify the relaxation process and the
resulting stationary state, potentially sustaining a
non-vanishing probability of mutual cooperation.

Persistent Cooperation: A non-vanishing stationary
probability of mutual cooperation, $P(CC) = \text{Tr}(|CC\rangle\langle CC| \rho_{\rm ss})$ that persists as $t \to \infty$.

Quantum-Classical Competition: {\it Assume the Detailed Balance Condition.} The stationary value of \(P(CC)\) depends on the competition between the cognitive frequency \(\Omega\) and the dissipative transition rates. In a simplified parametrization with a characteristic relaxation rate \(\Gamma\), this dependence may be expressed in terms of the ratio \(\Omega/\Gamma\)
If the cognitive frequency $\Omega$ is high relative to the relaxation rate $\Gamma$, the internal reasoning effectively overrides the payoff-driven environment.

Thus, the GKSL framework allows for the emergence of cooperation not just as a transient fluke, but as a stable, stationary phenomenon resulting from the non-trivial commutation relations between the agent's internal deliberation ($H$) and the external incentive structure ($L_{mn}$).

\subsection{The possibility of $|CD\rangle$ and $|DC\rangle$ jumps}

While the proposed Hamiltonian $H$ specifically couples the mutual states $|CC\rangle$ and $|DD\rangle$, the intermediate ``asymmetric'' states $|CD\rangle$ and $|DC\rangle$ remain critical to the dynamics through the dissipative term $\mathcal{D}(\rho)$. 

The transition rates $\gamma_{mn} = \Gamma e^{\frac{\beta}{2}(u_m - u_n)}$ are governed by the full payoff hierarchy $T > R > P > S$. Consequently, the system does not simply toggle between mutual cooperation and mutual defection. Instead, the environment exerts ``temptation'' pressures:

\begin{itemize}
    \item \textbf{The Sucker-Temptation Leak:} Even if $H$ populates the $|CC\rangle$ state, the jump operators $L_{CD,CC}$ and $L_{DC,CC}$ drive the system toward the asymmetric states because the temptation payoff $T$ exceeds the reward $R$.
    \item \textbf{The Rational Sink:} Once the system enters $|CD\rangle$ or $|DC\rangle$, the hierarchy $P > S$ ensures a final dissipative pull toward the Nash Equilibrium $|DD\rangle$.
\end{itemize}

he states $|CD\rangle$ and $|DC\rangle$ therefore provide
dissipative pathways through which population can leave the
coherent $|CC\rangle$--$|DD\rangle$ sector. These transitions
may contribute to the loss of coherence, depending on the
structure of the Lindblad operators. The Hamiltonian $H$ acts
as a ``shortcut'' or ``quantum bridge'' that allows the agent
to mentally bypass these intermediate states during
deliberation. The final observed probability of cooperation
depends on whether the decision is reached at a transient
time $t$ when the coherent oscillation toward $|CC\rangle$
is near its peak, or whether dissipative transitions through
the asymmetric states have already substantially reduced
the probability of cooperation.

Set
\[
|1\rangle=|CC\rangle,\quad
|2\rangle=|CD\rangle,\quad
|3\rangle=|DC\rangle,\quad
|4\rangle=|DD\rangle.
\]

\begin{proposition}[Stationary cooperation probability]
Let the Hamiltonian be
$
H=\Omega\bigl(|1\rangle\langle4|
             +|4\rangle\langle1|\bigr),
$
and assume that all directed transitions between distinct
basis states are present, with transition rates
\[
\gamma_{mn}
=\Gamma_0\exp\left(\frac{\beta}{2}(u_m-u_n)\right),
\qquad m\ne n.
\]
Here, $\gamma_{mn}$ denotes the transition rate from
$|n\rangle$ to $|m\rangle$, and $\Gamma_0>0$.

Let $\kappa_{14}>0$ denote the total damping rate of the
coherence $\rho_{14}$ induced by the dissipative part of
the GKSL generator, as specified by Theorem~1. Define
\[
x_j=\exp\left(\frac{\beta u_j}{2}\right),
\qquad
S=\sum_{j=1}^{4}x_j,
\qquad
Q=\sum_{j=1}^{4}x_j^2,
\]
and
$
k=\frac{2\Omega^2}{\kappa_{14}}.
$
Then the stationary probability of mutual cooperation is
\[
\boxed{
P(CC)_{\mathrm{ss}}
=
\frac{
x_1\left[
\Gamma_0 x_1 S+
k x_4(x_1+x_4)
\right]
}{
\Gamma_0 S Q+
k(x_1+x_4)(2x_1x_4+x_2^2+x_3^2)
}.
}
\]
\end{proposition}

\begin{proof}
Write $p_j=\rho_{jj}$ for the populations in the four
basis states. Since the Hamiltonian couples only
$|1\rangle$ and $|4\rangle$, the stationary equation
for the coherence $\rho_{14}$ is
$
0=-i\Omega(p_4-p_1)-\kappa_{14}\rho_{14}.
$
Consequently,
\[
\rho_{14}
=-\frac{i\Omega}{\kappa_{14}}(p_4-p_1).
\]
Substituting this expression into the Hamiltonian
contribution to the population equations gives
\[
\left.\dot p_1\right|_H=k(p_4-p_1),
\qquad
\left.\dot p_4\right|_H=k(p_1-p_4),
\]
where $k=2\Omega^2/\kappa_{14}$. Thus, at the level of the stationary populations,
the coherent coupling between $|1\rangle$ and $|4\rangle$
contributes an additional symmetric transition rate $k$
between these two states.

The dissipative transition rates can be written as
$
\gamma_{mn}=\Gamma_0\frac{x_m}{x_n}.
$
The stationary population equations therefore become
\[
0=\sum_{m\ne n}
\left(
\Gamma_0\frac{x_n}{x_m}p_m
-\Gamma_0\frac{x_m}{x_n}p_n
\right)
\]
for the four states, with the additional terms
$k(p_4-p_1)$ in the equation for $p_1$ and
$k(p_1-p_4)$ in the equation for $p_4$.

Solving the stationary population equations together with
the normalization condition
$
\sum_{j=1}^{4}p_j=1
$
gives
\[
\begin{aligned}
p_1&=
\frac{x_1[\Gamma_0x_1S+kx_4(x_1+x_4)]}{D},\\
p_2&=
\frac{x_2^2[\Gamma_0S+k(x_1+x_4)]}{D},\\
p_3&=
\frac{x_3^2[\Gamma_0S+k(x_1+x_4)]}{D},\\
p_4&=
\frac{x_4[\Gamma_0x_4S+kx_1(x_1+x_4)]}{D},
\end{aligned}
\]
where
$
D=\Gamma_0SQ+
k(x_1+x_4)(2x_1x_4+x_2^2+x_3^2).
$
In particular, since
\[
P(CC)_{\mathrm{ss}}
=\operatorname{Tr}(
|CC\rangle\langle CC|\rho_{\mathrm{ss}})
=p_1,
\]
we obtain the claimed expression, which is the payoff-dependent Gibbs distribution.
\end{proof}

\section{Concluding Remarks}

The integration of the GKSL master equation into QLM
provides a framework for describing the evolution of
cognitive states as open-system dynamics, shaped by
the interplay between internal processes and an
informational environment.

In this framework, Lindblad operators represent
contextual pressures and external constraints, whereas
the Hamiltonian describes internal cognitive processes.
This distinction is reflected in the difference between
Passive and Active Hamiltonians. While passive dynamics
drive the system toward equilibria determined by the
environment, an Active Hamiltonian captures internally
generated tendencies that may compete with these
pressures. Their interplay can produce stable outcomes
that depart from classical rational predictions, as
illustrated by the Prisoner's Dilemma.

The analysis of cognitive beats further shows how
competing dynamical modes can generate slow oscillations
in decision probabilities, potentially offering
empirical access to the temporal structure of deliberation
through methods such as eye-tracking and EEG.

Future work may extend this framework to collective
cognitive dynamics, synchronization, and non-Markovian
memory effects, providing further insight into the
dynamical structure of cognition.

\subsubsection*{Author Contributions} 
\textbf{M.A.:} Conceptualization, Methodology (the decision-making via decoherence), Formal Analysis (mathematical counterparts in open quantum system). 
\textbf{I.B.:} Formal Analysis (formulation and proof of Proposition 2), Investigation (coupling quantum beats theory to eye-tracking experiments in decision-making), Writing -- Review \& Editing (refining interpretations of mathematical results for QCDM and editing the final manuscript). 
\textbf{A.K.:} Conceptualization (coupling QCDM with open quantum systems theory), Writing -- Original Draft, Writing -- Review \& Editing.

\section*{Acknowledgments} 
I.B. is supported by the grant ``Quantum Emotions: Using the Quantum Formalism to explain temporal order effects in memory for emotional events'', UK Research and Innovation (UKRI) Cross Research Council Responsive Mode (CRCRM); 
A.K. was partially supported by the EU-grant CA21169 (DYNALIFE) and visiting fellowship at Tokyo University Of Science.

\section*{Appendix: Towards experimental verification of the beats phenomenon}

For $N=3$, beats are theoretically predicted within the framework of decision making through decoherence. The primary empirical study involving this framework for $N=3$ was reported in \cite{eye}, where eye-tracking \cite{} served as a signature for the convergence toward a steady state. These experiments revealed the irreducible presence of oscillations during deliberation—a characteristic naturally captured by the GKSL-based model. However, at the time of that study, the 'beat' phenomenon had not yet been formally addressed within QCDM. Furthermore, the obtained in \cite{eye} experimental dataset lacked the temporal resolution required to resolve these beats. The theoretical analysis of the beat phenomenon presented in the present paper is intended to stimulate further experimental research, specifically aimed at gathering high-density data to validate these oscillations and simulate decision-making through decoherence.

Within the framework of decision making through decoherence, beat phenomena may occur even when the internal Hamiltonian is negligible, \(H\approx 0\). In this scenario, oscillatory components arise not from periodic modulation of the Hamiltonian but from the interplay of multiple dissipative channels that generate complex eigenvalues in the Liouvillian spectrum. When at least two such oscillatory modes contribute to the observed probabilities, their interference may produce beat patterns.

\end{document}